\documentclass[review,sort&compress]{elsarticle}

\usepackage[utf8]{inputenc}
\usepackage{blindtext}
\usepackage{xcolor}
\usepackage{amsmath,amssymb,amsthm}
\usepackage[pagebackref,colorlinks=true]{hyperref}
\usepackage[noabbrev,nameinlink]{cleveref}
\usepackage{booktabs}
\usepackage{multirow}
\usepackage{tikz}
\usepackage{floatrow}

\newtheorem{remark}{Remark}
\newtheorem{proposition}{Proposition}

\graphicspath{{./fig/}}
\DeclareGraphicsExtensions{.eps,.pdf,.jpeg,.jpg,.png}

\usetikzlibrary{positioning, math}
\tikzset{
   node distance=8em, on grid,
   vertex/.style={circle, draw},
   every picture/.style={auto, thick},
}

\crefname{equation}{expression}{expressions}
\crefname{Equation}{Expression}{Expressions}

\usepackage{lineno,hyperref}
\modulolinenumbers[5]

\journal{Information Sciences}

\newcommand{\mc}[3]{\multicolumn{#1}{#2}{#3}}

\begin{document}

   \begin{frontmatter}
      \title{Extrema Analysis of Node Centrality\\in Weighted Networks}

      \author{Roger S. Passos}
      \ead{roger.passos@aluno.cefet-rj.br}
      \author{Douglas O. Cardoso\corref{ca}}
      \cortext[ca]{Corresponding author}
      \ead{douglas.cardoso@cefet-rj.br}
      \ead[url]{https://docardoso.github.io/}
      \address{Department of Computer Engineering,\\Federal Center for Technological Education Celso Suckow da Fonseca (CEFET-RJ), Petrópolis, RJ, Brazil}

      \begin{abstract}
         A very interesting matter of Network Science is assessing how complex a given network is.
         In other words, by how much does such a network departs from any general patterns which could be evoked for its description.
         Among other choices, these patterns can be defined in terms of node or edge properties, as one of the various modes of centrality.
         Although there are centrality indexes defined for weighted graphs, the discussions on this subject are far from over, especially because the influence of edge weights in this regard can vary not only in form but also in intensity, a still incipient approach so far.
         For the aforementioned complexity evaluation, the situation is even more acute:
         to the best of our knowledge, this was never addressed in the literature, from the standpoint of centrality in weighted graphs.
         This paper details in a colorful fashion a sound methodology covering both topics, as well as experiments confirming its practical applicability and future trends in this context.
      \end{abstract}


   \end{frontmatter}

   \section{Introduction \label{sec:intro}}
      Centrality assessment is a branch of study that emerged from complex network analysis and, in short, targets to indicate how important nodes of a network are to its structure.
      In practical terms, this can mean identifying influential people in online social applications~\cite{RiquelmeC16}, analyzing good routes in urban transportation networks~\cite{crucitti2006centrality}, or even understanding which proteins can be considered essential for protein-protein interactions in a cell~\cite{jeong2001lethality}.
      Three centrality measures are considered the most classic:
      degree, closeness, and betweenness.
      They were originally developed for networks in which relationships are binary, i.e., there is or not a connection between any pair of vertices~\cite{freeman1978centrality}.
      But there are networks with numerical weights attached to the connections, quantifying some characteristic of these:
      e.g., distance, intensity, capacity, or cost.
      Likewise, several centrality measures also have been proposed with the latter type of network in mind.

      Some of these~\cite{OpsahlAS10} were proposed with the objective of incorporating concomitantly the number of connections -- as originally intended in this regard -- and the weights associated with them, allowing to regulate the relative importance given to each of these attributes.
      This dual description is relevant since in graphs that represent internet traffic, for example, for paths that have the same distance or cost, the quality of resources flowing through more intermediate nodes tends to be worse than with fewer of them.
      The measure, then, allows setting this kind of situation.
      Moreover, there can be several other situations in which such a tuning makes centrality assessment more assertive.
      Thus, the present work sought to extend such previous developed adjustable centrality generalizations to weighted networks, trying to obtain more polished and secure understanding of the information they provide.

      This paper details multiple contributions to the literature in the field of the just mentioned subject.
      The first of which is the establishment of adjustable centrality measures based on the logarithm of the ratio between weighted and unweighted centralities, which possess quite interesting properties.
      This serves as the cornerstone to the evaluation of how wide is the range encompassing all possible node rankings defined for a perspective of centrality varying the importance of edge weights in this regard, which can be interpreted as proxy to the assessment of weights relevance to network structure.
      One last part of this whole is an algorithm to realize such a evaluation without relying on a brute-force approach, as well as a formal proof of its predicates.

      The following sections, which complete this research report, are enumerated and summarized next.
      \Cref{sec:basic} establishes the theoretical foundations on which this work was built.
      \Cref{sec:over} provides a brief systematization and analysis of previously proposed node centrality measures which allow edge weights influence on their computations to be calibrated.
      An alternative methodology to such indexes which overcomes some of their weakness is introduced in \cref{sec:improv}.
      While \cref{sec:assess} shows how this methodology can be extended to allow analyzing weighted networks from so far unexplored point of view, \cref{sec:inter} formally describes how this can be efficiently realized.
      \Cref{sec:eval} details the application of the herein proposed ideas in a variety of practical scenarios.
      At last, \cref{sec:end} brings some finishing remarks and points possible continuations.

   \section{Basic Concepts and Notation \label{sec:basic}}
      A (simple, undirected) \emph{graph} or \emph{network} $G=(V, E)$ is a structure defined by the elements of set $V$, called \emph{vertices} or \emph{nodes}, which may have pairwise relationships, named \emph{edges, ties} or \emph{links}, that compose the set $E \subseteq \binom{V}{2}$.
      A node $u$ is \emph{adjacent} to a node $v$ if there is the edge $e=\{u, v\} \in E$.
      A \emph{path} between nodes $u$ and $v$ is a non-repeating sequence of nodes starting with $u$ and ending with $v$, or vice-versa, in which every two neighboring entries in the sequence are adjacent nodes in the network.
      The number of edges traversed in a path is its \emph{length}.
      If there is a path between every pair of vertices, the graph is \emph{connected}.
      There are also \emph{weighted} or \emph{valued} graphs, whose edges have a numerical value associated to each of them.
      These can be denoted as $G=(V,E,w)$, where $w : \binom{V}{2} \to \mathbb{R}$.
      In numerous situations such values are strictly positive, what is also assumed in the remainder of this text.
      As a consequence, it can be conveniently assumed that if $e \not \in E$, then $w(e) = 0$.

      The \emph{degree} centrality refers to the number of edges incident upon a node:
      \begin{equation}\label{DC} k_u=C_D(u)=|\{e \in E : u \in e\}|\ . \end{equation}
      This index covers only the influence of the nodes from a local perspective on the graph.
      On the other hand, the \emph{closeness} centrality is the inverse of the total geodesic distance of one node to all the others (relying on the assumption that the graph is connected, which is not strictly necessary~\cite{WF1994}).
      Between two nodes $u$ and $v$ the distance $dist(u, v)$ is the length of the shortest path between $u$ and $v$.
      The closeness, then, is formalized as:
      \begin{equation}\label{CC} C_C(u)=\left[\sum\limits_{v \in V}dist(u,v)\right]^{-1}\ . \end{equation}
         
      The two Freeman’s measures~\cite{freeman1978centrality} just presented were originally proposed for unweighted graphs.
      Based on them some generalizations were developed with valued networks in mind.
      In this regard, Barrat et al.~\cite{barrat2004architecture} extended the notion of degree of a node to what was called \emph{strength}, defined by the sum of the weights of all edges incident in it:
      \begin{equation}\label{WDC} s_u=C_D^{w}(u)=\sum\limits_{\{u, v\} \in E} w(\{u, v\})\ . \end{equation}
      Newman~\cite{newman2001scientific} introduced the closeness to weighted networks based on Dijkstra's shortest path algorithm~\cite{Dijkstra59}.
      In this method, the length of a path is the sum of the weights of each tie in it.
      The minimum length of a path between nodes $u$ and $v$ taking weights into account is denoted by $dist_w(u, v)$.
      With the notion of distance aligned with this approach, the weighted closeness centrality is:
      \begin{equation}\label{WCC} C_C^{w}(u)=\left[\sum\limits_{v \in V} dist_w(u,v)\right]^{-1}\ . \end{equation}
      It should be noted that Dijkstra's algorithm was designed for graphs in which weights have a negative connotation, i.e., the shortest paths have the least accumulated weight.
      Thus, when the weights concern positive relations (e.g., affinity, capacity, similarity), their multiplicative inverses are commonly used for shortest path evaluation~\cite{OpsahlAS10}.
      
      To evaluate or compare non-trivial properties of networks, random graphs~\cite{Bollobas11} with specific characteristics can be used.
      A key model in this sense is the Erd\H os-Rényi (ER)~\cite{erdos59a}, whose notation $G(n,p)$ represents a graph that has $n$ vertices and $p$ is the probability of each pair of vertices to be connected by an edge.
      Although edge weights are not incorporated by the original ER model, it has been taken as the basis of others which include this feature.
      Given an ER graph, the possibly simplest way to obtain its weighted version is to assign to each link a value randomly defined according to an uniform probability distribution~\cite{kalisky2006scale}, for example.
      Another ER-based weighted model, with quite interesting properties, is the WRG~\cite{Garlaschelli_2009}, implemented by having the edge weights assigned according to a Geometric distribution which uses the value of $1-p$ as its parameter.
      
      The random models just referred to, especially the ER, are sometimes seen as lacking practical usefulness since their properties are not aligned with those observed in some remarkable real-world networks, considered as ``more complex''~\cite{kim2008complex}.
      Accordingly, alternatives that keep some features of a given network and randomize other aspects of it have been developed.
      A method that preserves node degrees while alters network topology is based on pairwise edge switching~\cite{milo2003uniform}.
      It consists of performing the following procedure a sufficient number of times:
      randomly choose two edges $\{a,b\}$ and $\{c,d\}$;
      then swap their ends to form the edges $\{a,d\}$ and $\{b,c\}$. 
      A weighted version of this idea was previously used~\cite{opsahl2008prominence}, which preserved edge weights distribution but not node strengths or network topology.
      Other models that only shuffle the edge weights, conserving topology~\cite{barrat2004architecture,opsahl2008prominence}, or that randomize node degrees and edge weights while the node strengths remain unchanged~\cite{Ansmann_2011} are also known.

   \section{An Overview of Adjustable Centrality Indexes \label{sec:over}}
      The study of topological properties of complex networks has evolved over time and is ultimately important because these attributes are responsible for defining their functionality.
      For example, the structure of a social network affects the propagation of information or diseases in it, and the topology of a electrical grid changes supply robustness and stability~\cite{strogatz2001exploring}.
      Often, individual nodes assume a key function in this regard (e.g., people with great influence, large power stations) and, consequently, it is desirable to identify these nodes.
      That is the purpose of the centrality measures~\cite{DasSP18}.
      The versions of degree and closeness centralities presented in the \cref{sec:basic} take into account a single feature of the connections to determine node centrality:
      the number of links or, leaving aside the former, the sum of edge weights.
         
      More recently, some centrality indices for weighted networks have approached this goal of incorporating both tie features simultaneously.
      Laplacian centrality~\cite{QiFWWZ12} and the WKPN algorithm~\cite{TangSDMDH20} have tried to merge local and global characteristics of the vertices to assess their importance.
      The former is based on 2-walks in which a node participates.
      The latter considers single-source shortest paths with $K$ edges, where the analyzed node is the spread source in a infectious disease model.
      The h-degree index~\cite{ZhaoRY11} of a node displays only a local perspective and is defined as the greatest value $d_h$ so that a node has at least $d_h$ edges and the weight of each of them is greater than or equal to $d_h$.

      Compared to the previous ones, the measures proposed by Opsahl et al.~\cite{OpsahlAS10} are more adaptable, in terms of balancing the relative importance given to edge weights and counts.
      Such a statement is based on their use of a tuning parameter $\alpha$.
      Aiming at a local perspective, the degree~\eqref{DC} and strength~\eqref{WDC} of a node are combined to define a new version of degree centrality as follows:
      \begin{equation}\label{GeomeanWDC} C_D^{prod,\alpha}(u)={k_u}^{(1-\alpha)} {s_u}^{\alpha}\ . \end{equation}
      A similar approach was later applied to the weighted k-shell decomposition method~\cite{GarasSH12}.
      In addition, some of the subsequent works include the analysis of node centrality in an online social network under the variation of $\alpha$~\cite{yustiawan2015degree}, and an attempt to find an optimal value for such a parameter~\cite{wei2012degree}.
      An alternative closeness centrality, however, is based on weighted distances calculated on the graph obtained from the original by raising each edge weight to $\alpha$, represented by $dist_\alpha(u, v)$:
      \begin{equation}\label{ExpsumWCC} C_C^{sum,\alpha}(u)=\left[\sum\limits_{v \in V} dist_\alpha(u,v)\right]^{-1}\ . \end{equation}
      For both measures, when $\alpha=0$, the same results of Freeman's original measures are produced, disregarding weights.
      Whereas when $\alpha=1$, the outcome corresponds to the generalizations for valued networks. 

      Besides the intrinsic differences regarding node reach in the network, Opsahl’s measures differ with respect to the type of summarization~\cite{BorgattiE06} each of them employ.
      In \cref{GeomeanWDC}, both degree and strength are calculated in their original form, and the tuning parameter used to offset the product of these values, similarly a geometric mean.
      In \cref{ExpsumWCC}, $\alpha$ is used to directly alter edge weights as a power to their exponentiation, which are then summed for shortest path computations, resembling the p-norm of a vector.
      Targeting to clarify such a discrepancy, these measures are herein identified with labels \emph{prod} (product) and \emph{sum}, instead of as originally denoted:
      $C_D^{w\alpha}(u)$ and $C_C^{w\alpha}(u)$, respectively.

      Moreover, such an approach favors considering other combinations of the available parts.
      Therefore alternatives to Opsahl's degree and closeness centralities could be defined as follows, respectively:
      \begin{align}
         C_D^{sum,\alpha}(u)  & = \sum\limits_{\{u, v\} \in E} [w(\{u, v\})]^{\alpha}\ ; \label{ExpsumWDC} \\
         C_C^{prod,\alpha}(u) & = {C_C(u)}^{(1-\alpha)} {C_C^{w}(u)}^{\alpha}\ .  \label{GeomeanWCC} 
      \end{align}
      Both pairs of adjustable degree and closeness centralities produce the same results when $\alpha$ is 0 or 1.
      \Cref{tab:measures_interpr} provides a qualitative description of the behavior of the measures when the parameter is outside the range defined by the just mentioned benchmark values.
      \begin{table}[ht]
         \small
         \begin{tabular}{c p{.4\textwidth} p{.4\textwidth}}
            \toprule
            \multirow{2}{*}{Centrality} & \multicolumn{2}{c}{Summarization} \\
            \cmidrule{2-3}
            & \multicolumn{1}{c}{Product} & \multicolumn{1}{c}{Sum} \\
            \midrule
            Degree
            &
            For $\alpha>1$, central nodes would have great aggregated edge weight resulting from a small number of ties.
            When $\alpha < 0$, the more edges and the lower the node strength, the better, which is useful when weights have a negative connotation.
            &
            Heavy edges are prioritized when $\alpha > 1$:
            as $\alpha$ tends to infinity, the nodes attached by the heaviest edges would become the most central ones.
            This preference is directly inverted for $\alpha < 0$, making lightest edges the most important.
            \\
            \hline
            Closeness 
            &
            For $\alpha > 1$, central vertices are those whose shortest paths to others generally have low accumulated weight but numerous edges. 
            With a negative parameter, the preferred vertices are whose distances to the others, disregarding the weights, are short, while discouraging the ones whose total weighted distance is small.
            &
            When $\alpha$ is negative, the central vertices would be those whose shortest paths to others have few, heavy edges.
            The opposite is true for $\alpha > 1$:
            shortest paths with numerous weaker ties are less affected in this scenario, favoring vertices whose shortest paths have such characteristics.
            \\
            \bottomrule
         \end{tabular}
         \caption{
            Effects on degree and closeness centrality assessment of \emph{prod} and \emph{sum} adjustment methods when $\alpha$ is out of the unit interval.
            \label{tab:measures_interpr}
         }
      \end{table}

      Now consider the graph presented in \cref{fig:opsahl_graph}.
      In order to exemplify the distinguishable behavior of measures with distinct types of summarization, \cref{fig:exp_measures} shows how the node rankings derived from \emph{prod} \eqref{GeomeanWDC} and \emph{sum} \eqref{ExpsumWDC} degree centralities vary according to the value assigned to $\alpha$.
      It is interesting to remark that the majority of rank changes occur outside the unit interval, suggesting that analyses restricted to it may not take full benefit of these indexes.
      And as a demonstration of their unique aspects, it can be seen that node $C$ leaves the last place as $\alpha$ grows in the \emph{prod} method -- while in \emph{sum} it occupies the same position in the entire interval -- indicating its greater accumulated strength, concerning the number of ties, compared to the others.
      Node $A$ in the \emph{sum} measure, however, occupies the second position at both extremes since its links have both the minimum and maximum edge weights.
      Despite their particularities, in specific conditions the measures may coincide, as indicated in \cref{prop:eq_prod_sum}.

      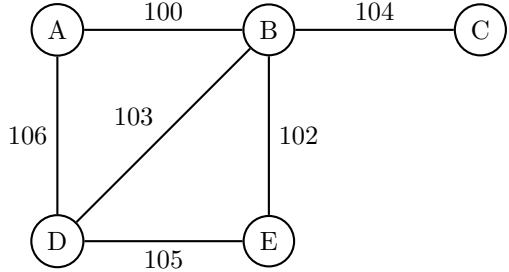
\begin{figure}[ht]
         \centering
         \begin{tikzpicture}
            \node [vertex] (A) {A};
            \node [vertex, right=of A] (B) {B};
            \node [vertex, right=of B] (C) {C};
            \node [vertex, below=of A] (D) {D};
            \node [vertex, below=of B] (E) {E};
            
            \draw (A) to node {100} (B);
            \draw (D) to node {106} (A);
            \draw (B) to node {104} (C);
            \draw (D) to node {103} (B);
            \draw (B) to node {102} (E);
            \draw (E) to node {105} (D);
            
         \end{tikzpicture}
         \caption{A sample graph. \label{fig:opsahl_graph}}
      \end{figure}
      
      \begin{figure}[ht]
         \centering
         \includegraphics{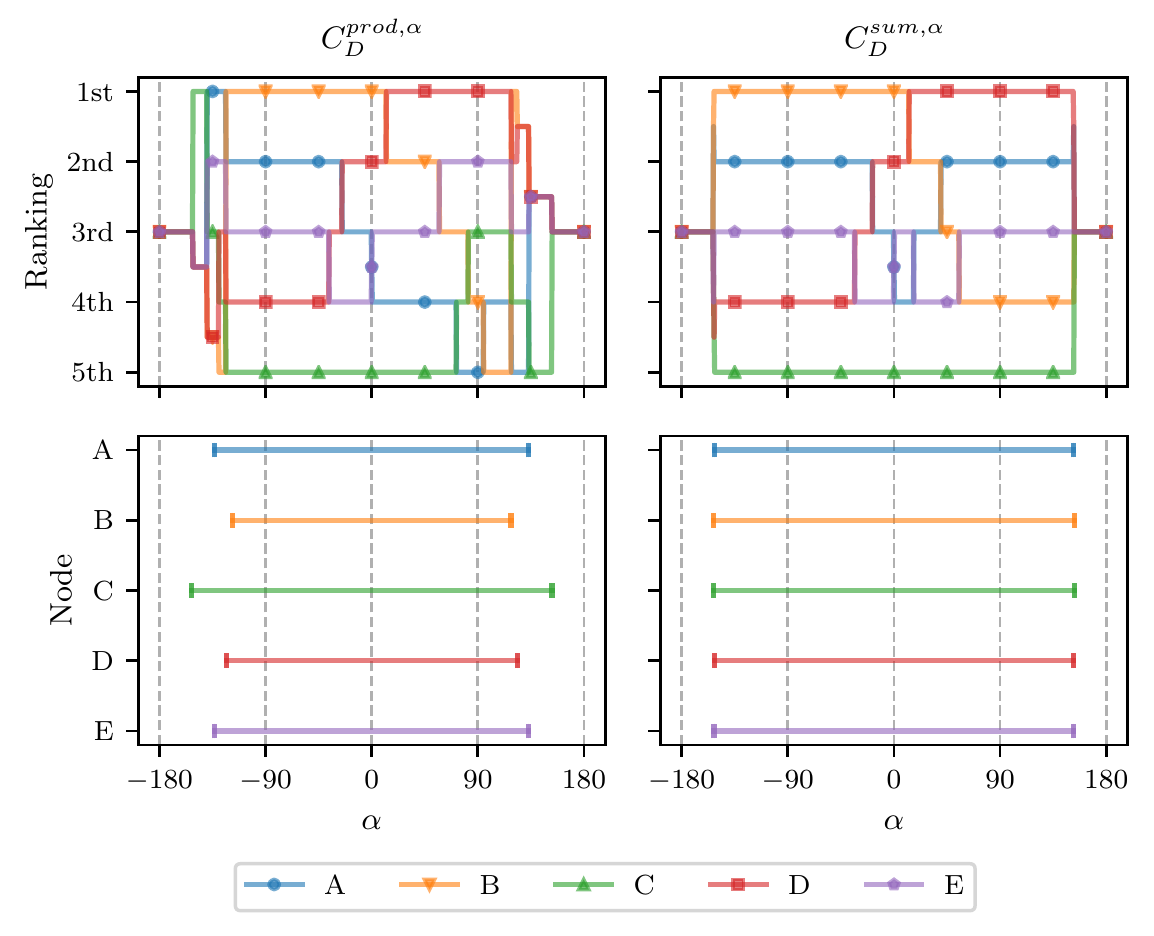}
         \caption{
            The first line displays the variation of the node ranks according to $\alpha$ for \emph{prod} and \emph{sum} alternatives of adjustable degree centrality.
            At the bottom, it is shown the range of $\alpha$ in which the measures are computable for each vertex.  
            Beyond these limits, the measures, and consequently the rankings, are miscalculated.
            At the extrema, all values under- or overflowed.
            The safe intervals for each measure are: $S_D^{prod} = [-118.03, 118.03],\ S_D^{sum} = [-152.2, 152.2]$.
            \label{fig:exp_measures}
         }
      \end{figure}
         
      \begin{proposition} \label{prop:eq_prod_sum}
         Given a weighted graph $G = (V, E, w)$ and a vertex $u \in V$, if $\forall \{u, v\} \in E,\ w(\{u, v\}) = t$, then $\forall \alpha \in \mathbb R,\ C_D^{sum,\alpha}(u) = C_D^{prod, \alpha}(u) = k_u t^\alpha$.  
      \end{proposition}

      \begin{proof}
         Since the weight of every link of $u$ is $t$, $s_u = k_u t$.
         And as the graph is assumed to be connected, $t = s_u/k_u$.
         Then, for any $\alpha \in \mathbb R$, 
         \begin{align*}
            C_D^{sum,\alpha}(u) & = 
            \sum_{\{u, v\} \in E} [w(\{u, v\})]^\alpha = 
            \sum_{\{u, v\} \in E} t^\alpha =
            k_u t^\alpha \\
            & = k_u \left (\frac{s_u}{k_u} \right) ^ \alpha =
            k_u^{(1-\alpha)} s_u^\alpha =
            C_D^{prod, \alpha}(u)\ .
         \end{align*}
      \end{proof}

   \section{Improvements on Adjustable Node Centrality \label{sec:improv}}
      Opsahl's measures have different methods of calculation from which it was presented a twofold description of adjustable degree and closeness centralities, which were subsequently compared.
      There is, however, a common thread in the four expressions:
      all of them are adjusted using exponentiation.
      For this reason they are notably susceptible to floating-point errors~\cite{Overton01,Goldberg91} such as overflow or underflow, possibly combined with loss of significance, thereby limiting the range of the $\alpha$ parameter in order to enforce numerical computability.
      To quantify that, consider the widely used IEEE 754~\cite{IEEE_standard85} double precision format, in which 11 of 64 bits are reserved for an exponent field to hold a value whose magnitude is not greater than 1024.
      As a consequence, real numbers whose absolute value is outside the interval $[2^{-1024}, 2^{1024}]$ cannot be properly represented in this format.

      Hence, \emph{safe intervals} of $\alpha$ can be established for the measures:
      let $G = (V, E, w)$ be an hypothetical graph;
      let $B$ be the set of values which could be bases for exponentiation by $\alpha$ during the computation of one of the centrality indexes for vertices of G;
      let $b = \max \{(\min B)^{\ -1},\ \max B\}$;
      then such a interval is defined as $S = \left[\log_b 2^{-1024},\ \log_b 2^{1024} \right]$.
      The definition of $B$ goes as follows:
      for $C_D^{prod,\alpha}$, $B = \{k_u : u \in V\} \cup \{s_u : u \in V\}$;
      for $C_D^{sum,\alpha}$ and $C_C^{sum,\alpha}$, $B = \{w(e) : e \in E\}$;
      and for $C_C^{prod,\alpha}$, $B = \{C_C(u) : u \in V\} \cup \{C_C^w(u) : u \in V\}$.
      At last, as a slight abuse of notation, hereinafter the following expression is used to denote the length of a real interval, like $S$:
      $|S| = \max S - \min S$.

      The method presented depends on the floating-point model employed, but its utility remains as long as arithmetic precision is finite.
      Moreover, it is important to notice that setting $\alpha$ respecting its safe interval can be seen as a necessary but not sufficient condition for numerical computability:
      after all, the values exponentiated by $\alpha$ are then used in other arithmetic operations.

      Although it is possible to determine safety parameters, such a range may be quite narrow, up to the point that its consideration disables $\alpha$ flexibility and, consequently, usefulness.
      A better approach in this sense can be obtained by taking the logarithm of the indices adjustable via product.
      This leads to a third pair of measures, defined as follows:
      \begin{align}
         C_D^{log,\alpha}(u)  & = \log{\left( \frac{s_u}{k_u} \right)} \cdot \alpha + \log{k_u}\ ; \label{LogmeanWDC} \\[4pt]
         C_C^{log,\alpha}(u)  & = \log{\left( \frac{C_C^{w}(u)}{C_C(u)} \right)} \cdot \alpha + \log{C_C(u)}\ .  \label{LogmeanWCC} 
      \end{align}

      Such an approach preserves the benchmarks for $\alpha$, although the classical values of degree and closeness come to be in logarithmic fashion instead.
      Furthermore, comparability of the centralities remains unchanged, as the logarithm is a monotonic function. 
      So the standings in~\cref{fig:exp_measures} for the \emph{prod} and \emph{log} methods are the same.
      This also reverberates on~\cref{fig:log_example}:
      due to the linear variation, the centrality is equalized across the parameter interval, making the entire ranking and proximity of quantities easily discernible.
      Another relevant feature of the \emph{log}-based measures is the enhanced flexibility for $\alpha$, since safe intervals become relatively unrestricted compared to those of exponential indices.
      \begin{figure}[ht]
         \centering
         \includegraphics{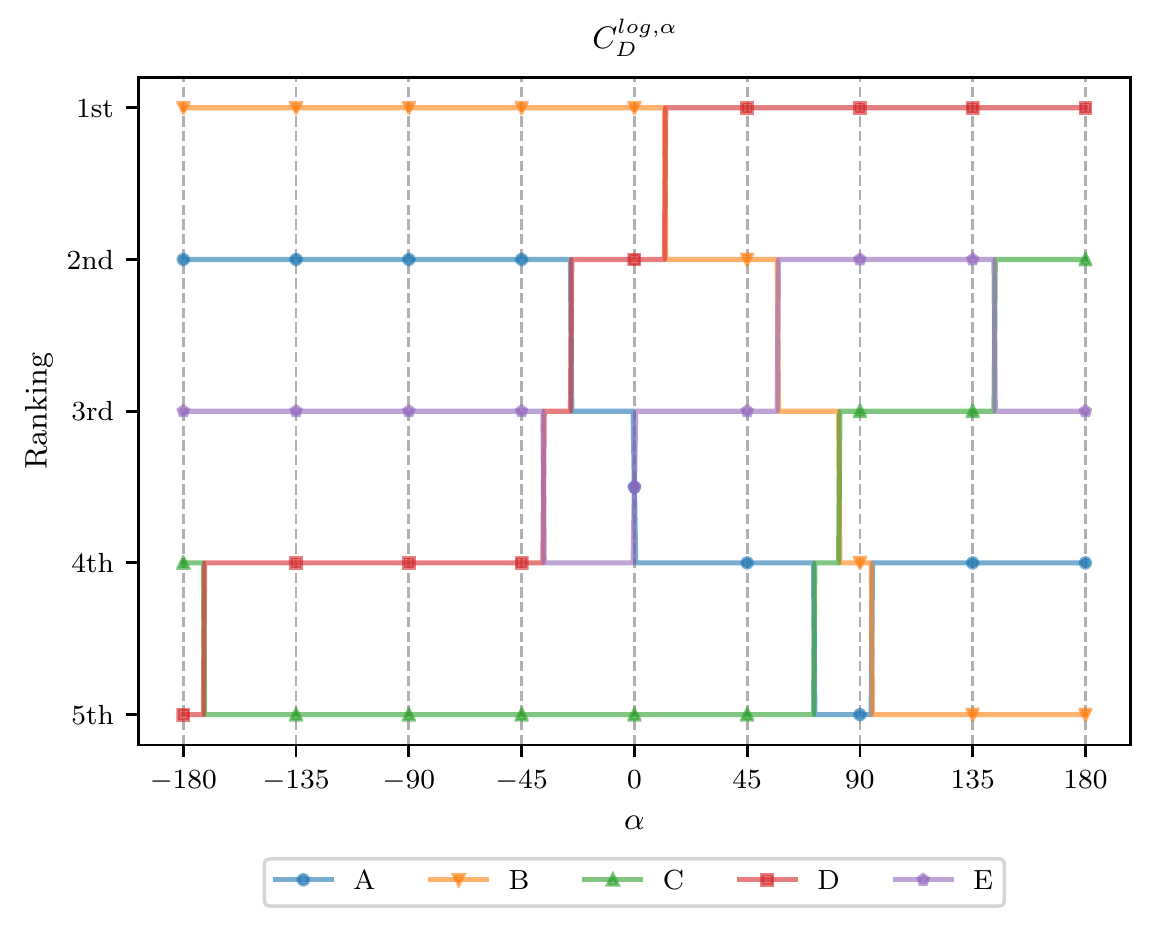}
         \caption{
            Variation of the node ranks according to $\alpha$ for \emph{log} alternative of adjustable degree centrality.
            The exchanges which happen when $\alpha \not \in [-135, 135]$ can now be correctly perceived, what was impossible in \cref{fig:exp_measures}.
            \label{fig:log_example}
         }
      \end{figure}

   \section{Assessment of Edge Weights Significance \label{sec:assess}}
      \Cref{sec:improv} introduced the idea of safe intervals for parameter $\alpha$, aiming at assuring numerical computability of the centrality indexes in question, and showed how the proposed measures $C_D^{log,\alpha}$ and $C_C^{log,\alpha}$ are superior to its counterparts in this regard.
      However in practice the entire safe intervals may not be functional:
      for example, there can be upper bound on $\alpha$ after which the node centrality ranking defined according to $C_D^{log, \alpha}$ remains unaltered;
      if the smallest value for such a bound is known, setting $\alpha$ beyond it can be considered senseless.
      This reasoning enables the definition of a \emph{useful interval}, denoted as $U$, according to the first and last change points of the ranking.
      This section details how these intervals can be interpreted, and related to some well-known network global measures.

      How these useful intervals can be understood is the first topic to be addressed because it also provides some motivation for their establishment.
      Varying parameter $\alpha$ allows to observe the same graph from different points of view, which can be materialized in the form of node centrality rankings, as already illustrated in \cref{fig:exp_measures}.
      Therefore the most straightforward application of useful intervals is to avoid looking for such alternative perspectives where there is none.
      But a more interesting possibility comes from evaluating how narrow the interval is, as an estimate how condensed are these perspectives, and how smooth is the transition between them.
      Since these transitions are solely due structural distinctions between the nodes which are reflected by benchmark centralities (i.e., when $\alpha \in \{0, 1\}$), the interval length can be seen as an indicator of node uniqueness in this sense.

      For a better comprehension of the information provided by interval length, consider the following analyses of two degenerate cases:
      \begin{itemize}
         \item
            Let $G = (V, E, w)$ be an hypothetical graph, such that $\forall e \in E,\ w(e) = 1$ and $\exists \{u, v\} \subset V,\ k_u \not = k_v$.
            In this scenario, $\forall u \in V,\ C_D(u) = C_D^w(u) = k_u$ and, therefore, $\forall u \in V,\ \forall \alpha \in \mathbb R,\ C_D^{log, \alpha}(u) = \log k_u$.
            In other words, the node ranking with respect to degree centrality would be the same regardless of $\alpha$. 
            Consequently $U_D = (-\infty, +\infty) = \mathbb R$, as there is no upper or lower bound on the value of $\alpha$ for delimiting changes in such a ranking.
            This reflects how similar the nodes are considering the influence of edge weights on their centrality statuses.

         \item
            Now let $G = (V, E, w)$ be another hypothetical graph, such that $\forall \{u, v\} \subset V,\ k_u = k_v$ and $\exists \{u, v\} \subset V,\ s_u \not = s_v$.
            In this case, the only change point in node ranking with respect to degree centrality is when $\alpha = 0$, what leads to determining $U_D = [0, 0]$.
            Since all nodes have the same degree but not the same strength, the range of $\alpha$ which allows multiple perspectives of the graph is minimum, as the importance of edge weights on switching between these perspectives is maximum.

      \end{itemize}

      How wide the useful intervals are depends on network global structure, and provides a feedback which is global as well.
      The relevance of edge weights to such a structure was previously ruled as of lesser importance relatively, considering a comparison between nodes degrees and strengths~\cite{Mastrandrea_2014}.
      Despite the soundness of such a statement, having the means to concretely evaluate this is quite advantageous, since there is also the indication that such a difference in informativeness varies according to the network in question. 
      As an example, Small-world-ness~\cite{journal.pone.0002051}, another graph global measure, was developed with the same target of quantifying an until then imprecise concept.
      However, Small-world-ness as well as the majority of the most popular graph global measures~\cite{Watts1998Collective, ZhouM04, PhysRevE.67.026126} are unrelated to edge weights, leaving a blank this work targets to fill to some extent.

   \section{Extrema Intersection of Unbounded Lines \label{sec:inter}}
      As previously stated, a useful interval of $\alpha$ for a certain centrality index is determined by the lowest and highest values of $\alpha$ which lead to changes in the node ranking defined according to that index.
      Such changes can be easily identified in a diagram like \cref{fig:log_example} since they happen whenever there is an intersection between any pair of non-coincident lines.
      Therefore determining useful intervals for the log-based measures comes down to determining the ``leftmost'' and ``rightmost'' single-point intersections of a collection of unbounded lines.

      The most naive approach for such a goal is to exhaustively verify if the intersection of each pair of lines is one of the extrema, which has a time complexity of $O(n^2)$.
      Fortunately improving this cost is possible avoiding to determine all intersections but directly finding those of interest.
      The relatable problem of deciding if there is at least one intersection between any pair of line segments given a set of these was previously addressed this way, so that an $O(n \log n)$ solution was provided~\cite{ShamosH76}.
      Although this method cannot be directly adapted to the problem at hand, a solution with similar cost, principles, and applicability, which is not limited to the present context, could be found.
      It is based on the following statements:

      \begin{remark}
         There is a 1-to-1 correspondence between a line described as $y = a x + b$ and an ordered pair $(a, b)$, so that an indexed collection of lines can be denoted by $L = \{(a_i, b_i)\}$, with $l_i = (a_i, b_i)$ representing the $i$-th line.
      \end{remark}
      
      \begin{remark} \label{rmk:hasint}
         Given the lines $L = \{(a_i, b_i)\}$, if there is no intersection between the $i$-th line and any other line (i.e., every line has the same slope of $l_i$, being parallel to it), there is no intersection between any pair of lines, since all of them are pairwise parallel.
      \end{remark}

      \begin{remark} \label{rmk:toleft}
         Given the lines $L = \{(a_i, b_i)\}$, if two of them intersect when $x = r$, for some $r \in \mathbb R$, the leftmost intersection between elements of L happens when $x = s < r+1$.
      \end{remark}
      
      \begin{proposition} \label{prop:toleft}
         Let $(a_i, b_i)$ and $(a_j, b_j)$ be lines such that $a_i r + b_i > a_j r + b_j$, for some $r \in \mathbb R$.
         If their intersection happens when $x = s < r$, then $a_i > a_j$.
      \end{proposition}

      \begin{proof}
         Suppose that the antecedent of the desired conclusion is true.
         Since $s < r$, then $r = s + m$, with $m \in \mathbb R_{> 0}$.
         Therefore, $a_i r + b_i > a_j r + b_j \implies a_i (s + m) + b_i > a_j (s + m) + b_j$.
         From the supposition, $a_i s + b_i = a_j s + b_j$, what allows to derive from the last inequality that $ a_im > a_jm \implies a_i > a_j$.
      \end{proof}

      \begin{proposition} \label{prop:triplet}
         Let $l_i = (a_i, b_i), l_j = (a_j, b_j)$ and $l_k = (a_k, b_k)$ be lines such that $a_i > a_j > a_k$.
         Let $y_i = a_i o + b_i, y_j = a_j o + b_j$, and $y_k = a_k o + b_k$, for some $o \in \mathbb R$, with $y_i \ge y_j$ and $y_i \ge y_k$.
         Consider that the intersections between the lines $l_i$ and $l_j$, $l_i$ and $l_k$, and $l_j$ and $l_k$ happen respectively when $x = r$, $x = s$, and $x = t$.
         If $s \le r$, then $t \le s$.
      \end{proposition}
      
      \begin{proof}
         From the premisses, $$r = \frac{b_j - b_i}{a_i - a_j},\qquad s = \frac{b_k - b_i}{a_i - a_k},\qquad t = \frac{b_k - b_j}{a_j - a_k}\ .$$
         Since $a_i > a_j > a_k$, then $a_j = a_k + m$ and $a_i = a_k + m + n$, with $\{m, n\} \subset \mathbb R_{> 0}$.
         Now suppose that, $s \le r$.
         As $s - o \le r - o$, it can be stated that $$\frac{y_k - y_i}{m + n} \le \frac{y_j - y_i} n\ .$$
         From the fact that $$\frac{y_k - y_i}{m + n} \ge \frac{y_k - y_i} n\ ,$$ it follows that $$\frac{y_k - y_i} n \le \frac{y_j - y_i} n \implies y_k - y_i \le y_j - y_i \implies y_k \le y_j\ .$$
         This enables to state that $y_i \ge y_j \ge y_k$ and, consequently, that $y_j = y_k + p$ and $y_i = y_k + p + q$, with $\{p, q\} \subset \mathbb R_{\ge 0}$.
         Therefore the supposed inequality can be once again rewritten, this time as $$\frac{-(p+q)}{m+n} \le \frac{-q} n\ .$$
         Consequently, $$-pn - qn \le -qm - qn \implies -pm - pn \le -pm - qm \implies \frac{-p} m \le \frac{-(p + q)}{m+n}\ .$$
         Since the left- and right-hand sides of this last inequality respectively equal to $t$ and $s$, the proof is concluded.
      \end{proof}

      As a preamble to the core procedure, for the input lines $L = \{(a_i, b_i)\}$, consider that an intersection between any pair of lines is determined on $x = r$:
      as \cref{rmk:hasint} indicates, if this is not possible there is no intersection at all, and this can be decided in linear time.
      Next the lines are sorted in ascending order of $a_i (r + 1) + b_i$, so that the leftmost intersection happens on some $s < r+1$, as affirmed in \cref{rmk:toleft}.
      These two steps have a time complexity of $O(n \log n)$.

      At last, the algorithm itself has an inductive design~\cite{Manber88}, linearly incrementing on the lines.
      From now on the value of $x$ of the leftmost intersection between two of the first $n$ lines is represented by $x_n$.
      In the base case, only the first line is taken into account:
      in this trivial scenario, there are no intersections, what is denoted by $x_1 = -\infty$.
      This part has constant computational cost.

      Now, for the inductive step, first assume that $x_{n-1}$ is known.
      This can also be the solution for the first $n$ lines (i.e., $x_n = x_{n-1}$), but maybe an intersection between $l_n$ and any of the previous lines leads to a better solution.
      According to \cref{prop:toleft}, only lines whose slope is lower than that of $l_n$ could allow this improvement.
      And based on \cref{prop:triplet}, if there is more than one line matching this criterion, the one with the greatest slope, represented as $l_{n'}$, should be responsible for the leftmost intersection with $l_n$, or no intersection with this line would be the leftmost of all.
      Thus, considering that an intersection between $l_n$ and a hypothetical $l_{n'}$ happens when $x = t$, it follows that
      $$
         x_n = 
         \begin{cases}
            x_{n-1} & \text{if the slope of $l_n$ is the smallest of the first $n$ lines} \\
            t & \text{if } x_{n-1} = -\infty, \\
            \min(x_{n-1}, t) & \text{otherwise}.
         \end{cases}
      $$

      With respect to the algorithmic complexity, each iteration of the inductive step is dominated by the identification of $l_{n'}$.
      Any self-balancing binary search tree~\cite{Knuth98a} could be used in this regard, whose queries and updates require $O(\log n)$ elementary operations.
      Consequently, the entire procedure has a time complexity of $O(n \log n)$ and a space complexity of $O(n)$.

      As a closing remark, it is noticeable that the explanation just provided focused on finding the leftmost intersection of a pair of lines of $L$.
      However, since the goal is to determine the length of the intervel in which all intersections happen, only this is now enough.
      Fortunately the rightmost intersection can be determined by the same means, using the symmetry argument demonstrated next in \cref{prop:reciprocal}.
      \begin{proposition} \label{prop:reciprocal}
         The rightmost intersection of lines $L = \{(a_i, b_i)\}$ is the reciprocal leftmost intersection of lines $L' = \{(-a_i, b_i)\}$.
      \end{proposition}
      
      \begin{proof}
         Suppose that $i$-th and $j$-th lines of $L$ are those whose intersection is the rightmost one, which happens on $x = r$.
         Then, $a_i r + b_i = a_j r + b_j \implies (-a_i) (-r) + b_i = (-a_j) (-r) + b_j$.
         Therefore there would also be an intersection between lines $(-a_i, b_i)$ and $(-a_j, b_j)$ of $L'$ on $x = -r$.
         Now suppose, for the sake of contradiction, that this is not the leftmost intersection of $L'$, but that of lines $(-a_k, b_k)$ and $(-a_l, b_l)$ on $x = -s < -r$.
         Then there is also an intersection between $(a_k, b_k)$ and $(a_l, b_l)$ of $L$ on $x = s > r$, what opposes the initial assumption.
         This concludes the proof.
      \end{proof}

   \section{Experimental Evaluation \label{sec:eval}}
      This section is focused on depicting meaningful examples of the application of the measures proposed in~\cref{sec:assess,sec:improv}.
      For this purpose the null models described in~\cref{sec:basic} were used, considering the variation of parameters such as the number of nodes and edges, and the weights distribution.
      The same evaluation was also performed for real-world networks and respective randomly generated surrogates which preserve some attributes of the original artifacts.
      The multiplicative inverse of the edge values was used to calculate the weighted distances in all cases.

      It is worth pointing out that the experiments and respective analyses are far from exhausting the proposed methodology.
      Here it is shown how the measures behave in various scenarios, aiming at providing insights into how some graph attributes can influence them.
      Further studies may consider other characteristics, in addition to those assessed here, for the purpose of elucidating more properties and accomplishing other predictions.

      \subsection{Synthetic Networks}
         Regarding the random graph models, it was implemented the ER, $G(n,p)$, but with weights assigned from a normal distribution:
         henceforth this is called ER+normal model.
         This derivative was adopted because it is very convenient to verify the effects of the respective mean and variance on our measures.
         The WRG model~\cite{Garlaschelli_2009}, on the other hand, was previously established in the literature in a clear fashion, producing graphs whose edge weights follow a geometric distribution.
         Such synthetic data favor results interpretability since factors which could affect them are known a priori.

         The default values used for each parameter are the following:
         \emph{Number of vertices (n)}, 200;
         \emph{Probability of connection (p)}, 0.2;
         \emph{Mean weight ($\mu$)}, 10;
         \emph{Standard deviation of weights ($\sigma$)}, 1.
         For each parameter setting, 1000 random graphs were generated for statistics computation:
         for each graph, averages of measures of interest were calculated and then summarized to produce the figures shown next.
         Since the length of a useful interval can be infinite, the median and interquartile range (IQR) were used to summarize some results since they are superior to the counterparts mean and standard deviation with respect to handling such a special value.
        
         In the first experiment to be reported, the effects resulting from varying the parameters of the normal distribution for ER+normal graphs were analyzed.
         Since topology settings remained unaltered, only measures based on the weights would reflect such changes.
         Hence, as shown in the first and third columns of~\cref{fig:normalpar_chg}, the raise of the mean edge value was naturally reproduced by the node strengths.
         Likewise, the average weighted distance decreased once the weights were larger -- indicating closer relationships.
         And as confirmed in the second column of the figure, with the increase of $\sigma$ it was expected that, on average, strengths would not vary while distances would decrease slightly.
         That is because, with a greater variety of weights, it is more likely that the minimum weight of all paths between any two vertices becomes smaller.
         \begin{figure}[ht]
            \centering
            \includegraphics{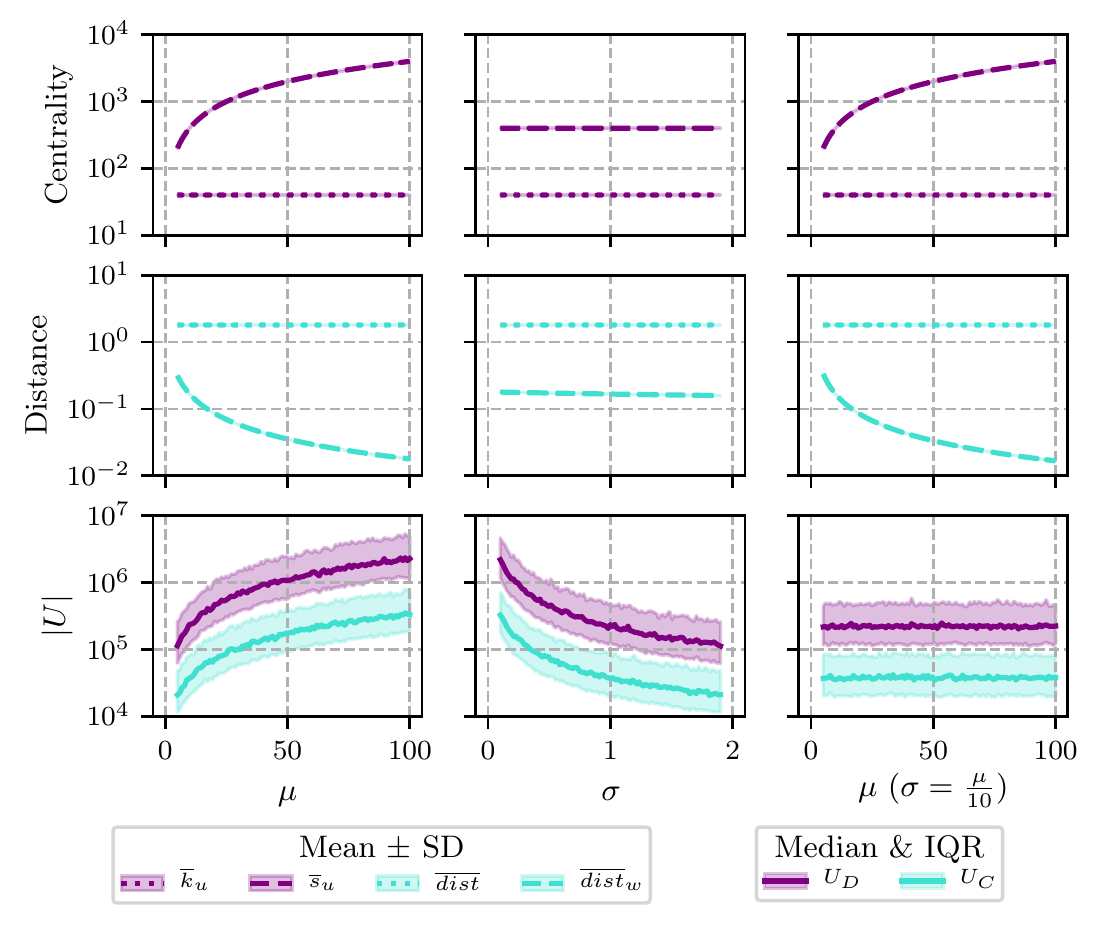}
            \caption{
               The effects that the variation of the normal distribution parameters cause on edge weights significance.
               \label{fig:normalpar_chg}
            }
         \end{figure}

         Regarding the lengths of the useful intervals, $|U_D|$ and $|U_C|$, it is remarkable that their variation always followed the inverse trend of the weights coefficient of variation (CV), $\sigma / \mu$.
         When only the mean increased, the CV decreased while both $|U_D|$ and $|U_C|$ also increased, suggesting a smaller relevance of the weights for both centrality perspectives.
         On the other hand, when the standard deviation of the weights increased while the mean remained constant, the opposite happened for both CV as well as the useful intervals.
         At last, when both mean and standard deviation varied together, it can be noticed the stability of the importance of the weights, as shown in third column of \cref{fig:normalpar_chg}.

        
         As another experiment, it was also considered the consequences of changing the number of nodes, $n$, for ER+normal graphs as well as for realizations of the WRG model.
         In this scenario the probability of connection was $p = (\ln n) / \sqrt n$, which is greater than the asymptotic threshold of connectivity of $G(n,p)$, $p = (\ln n) / n$~\cite{rgd}, for $n > 4$.
         \Cref{fig:size_chg} shows that, although a larger $n$ leads to a smaller $p$, the expected number of neighbors grows with $n$:
         $\mathbb E[k_u] = (n-1)p \simeq (\ln n) \sqrt n$.
         And regarding the average unweighted distance, its growth with $n$ was anticipated~\cite{PhysRevE.70.056110}.
         \begin{figure}[ht]
            \centering
            \includegraphics{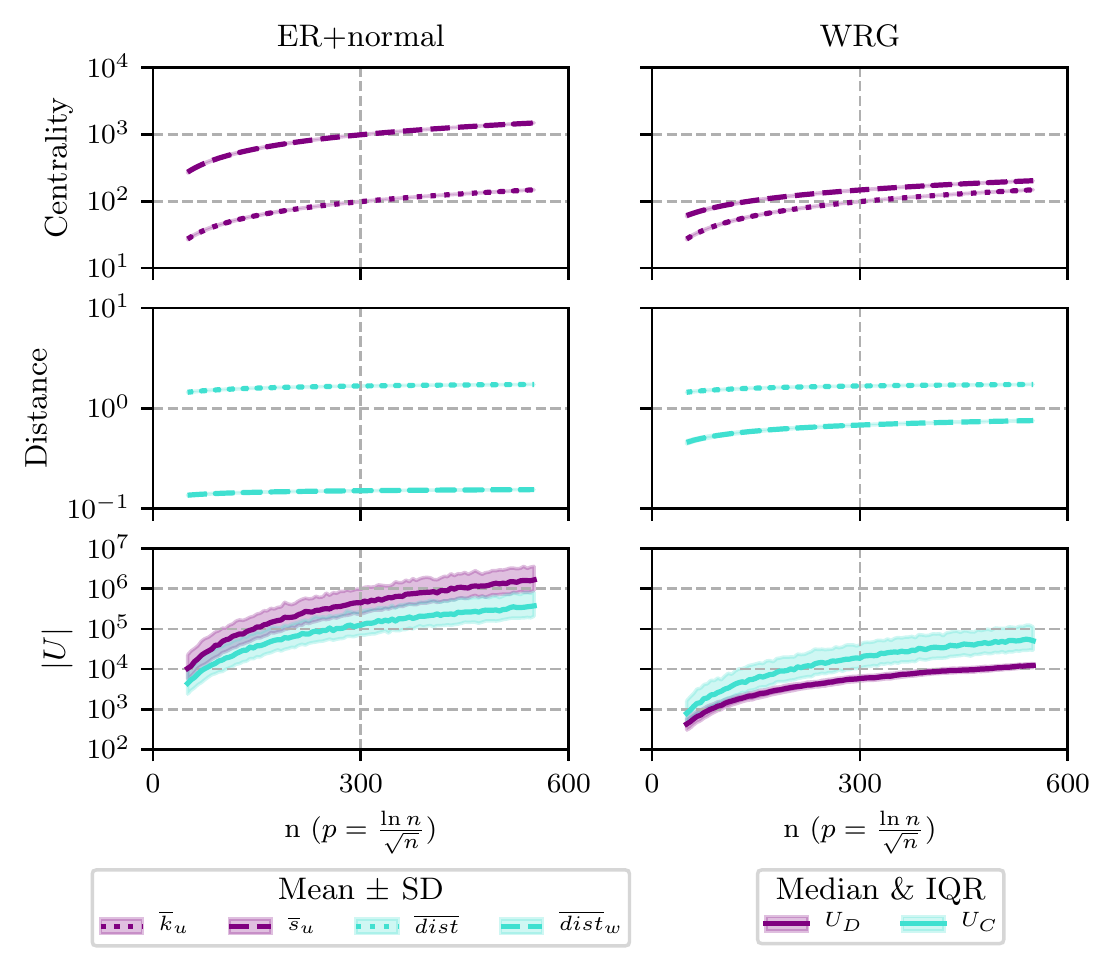}
            \caption{
               Variation of the graph size and its influence on the properties of random graphs obtained from the ER+normal and WRG models.
               \label{fig:size_chg}
            }
         \end{figure}

         The expected strength is $\mathbb E[s_u] = \mathbb E[k_u] \mathbb E[w]$, with the expected edge weight, $\mathbb E[w]$, being defined in a distinct fashion for each of the two graph models:
         for the ER+normal, $w \sim \text{Normal}(10, 1) \Rightarrow \mathbb E[w] = 10$, and since $\mathbb E[w]$ is constant, $\mathbb E[s_u]$ is directly proportional to $\mathbb E[k_u]$;
         for the WRG, $w \sim \text{Geometric}(1-p) \Rightarrow \mathbb E[w] = 1/(1-p)$, so that the mean weight decreases as $n$ increases, but without exceeding the pace of growth of $\mathbb E[k_u]$, resulting in a still increasing $\mathbb E[s_u]$.
         The average weighted distance for the ER+normal model differs from the unweighted counterpart by a constant multiplicative factor, $\mu = 10$.
         And as in the WRG $\mathbb E[w]$ tends to 1 with the growth of $n$, the weighted distances approach the unweighted ones.

         It would be reasonable to predict that the increase of the number of nodes would make the occurrence of nodes with similar centrality profiles more probable, prompting the enlargement of the useful intervals.
         Lastly, the lower values of $|U_D|$ and $|U_C|$ for the WRG, in comparison with those of the ER+normal, can be associated with the fact that the former's coefficient of variation of weights is always higher than the latter -- whose value is constant, 0.1.
         This would result in a greater diversity of edge values and, consequently, of centralities.

        
         Regarding the probability of connection, as can be seen in~\cref{fig:prob_chg}, it was varied from a relatively small value, discarding generated graphs which were not connected, up to a setting that approximates that of a complete graph.
         As $p$ increases, the node degrees also increase, while the unweighted distances decrease due to the higher amount of edges.
         For the ER+normal once again the strengths and weighted distances averages differ from their unweighted versions by the constant $\mu = 10$.
         Whereas for the WRG, the average strength is aligned to the fact that $s_u \sim \text{NegativeBinomial}(n-1, p) \Rightarrow \mathbb E[s_u] = p(n-1)/(1-p)$, while the average weighted distance exhibit a similar pattern, but vertically reflected, with a decreasing tendency:
         this could be owed to the use of the multiplicative inverse of weights for distance calculations.
        
         Observing the length of the useful intervals, the discrepancy between $|U_D|$ for each the considered graph models is remarkable.
         When $p$ is closer to 0, in either ER+normal and WRG the mean and standard deviation of both $s_u$ and $k_u$ have their lowest values.
         Despite that, for the first, $|U_D|$ has its maximum value, while for the second, its value is minimum, so that the decreasing tendency of $|U_D|$ for the ER+normal contrasts with the opposite behavior respective to the WRG.
         This can be interpreted as an evidence of how distinct is the information provided by the useful intervals compared to that of the centrality measures, regardless of the weights.
         Meanwhile, $|U_C|$ has a relatively similar behavior for both methods.
         This occurs, however, even with the sample CV of $dist_w$ for the WRG being consistently greater that of ER+normal for most of the range of $p$.
        
         \begin{figure}[ht]
            \centering
            \includegraphics{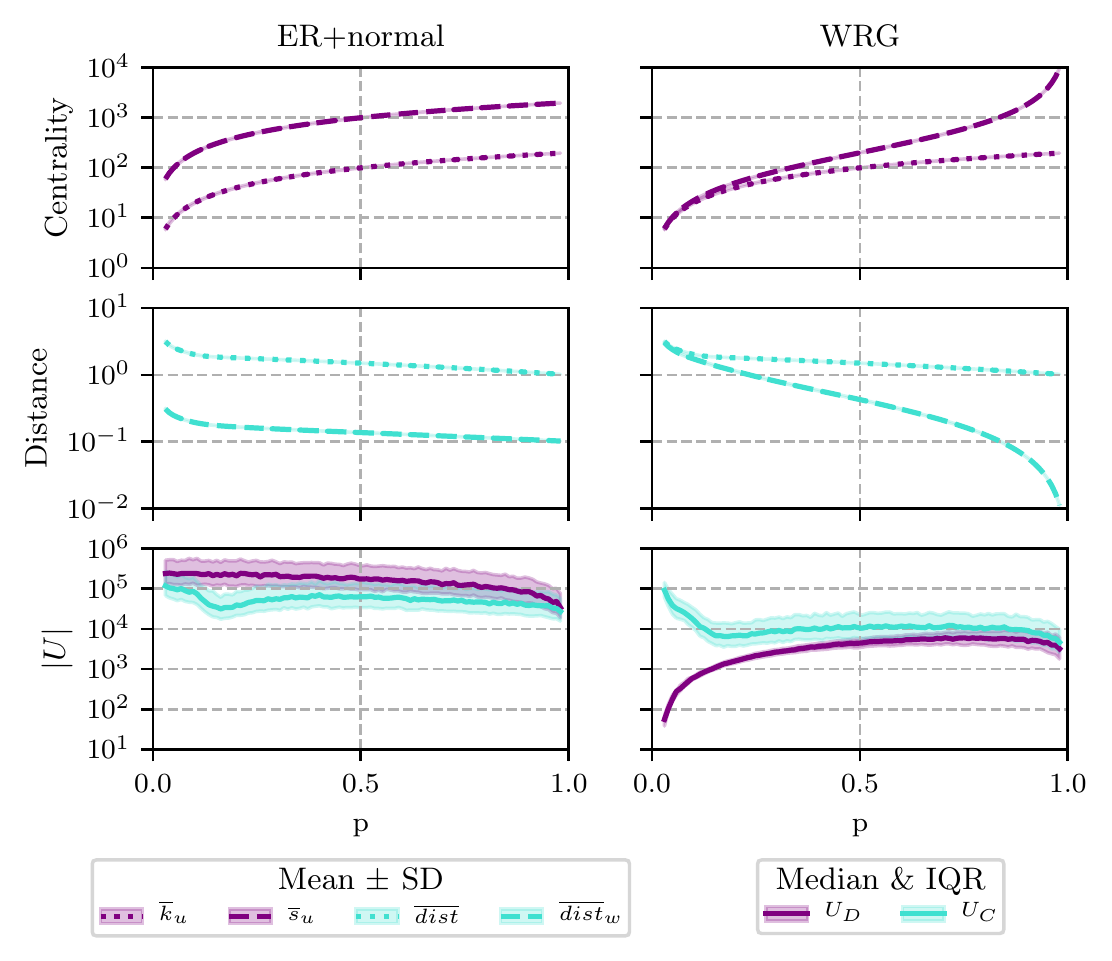}
            \caption{
               Variation of the probability of connection and its influence on the properties of random graphs obtained from the ER+normal and WRG models.
               \label{fig:prob_chg}
            }
         \end{figure}

      \subsection{Real Networks}
         The three datasets/networks briefly described below were also analyzed.
         A summary of some relevant features of each network is presented in~\cref{tab:realdatasets}.
         \begin{itemize}
            \item \emph{Freeman's EIES Network}~\cite{freeman1979networkers}:
               established in 1978, the dataset encompasses three different networks from which we opted to use the one that represents the number of messages sent among 32 researchers using an electronic communication tool.
               This network was also used in Opsahl's original work~\cite{OpsahlAS10}.
               Originally, its links are directed.
               We chose to transform the network into an undirected one, in which each edge $e=\{u,v\}$ has a weight equal to the sum of the values its directed versions in the original graph: $w(e)=w((u,v))+w((v,u))$.
               Self-loops were also removed.
            \item \emph{US Air Transportation Network}~\cite{colizza2007reaction}:
               the nodes represent the top 500 US airports according to the amount of traffic in 2002, while the weights consist of the number of seats available on scheduled flights between airports.
            \item \emph{Brazilian Congressmen Voting Network}~\cite{cnaa006}:
               represents the similarity between voting profiles of Brazilian congressmen from 2015 to 2016.
               Each node represents one congressman and an edge will be present if two congressmen voted identically on at least one bill.
               The weights range from 0 (no commonalities) to 1 (identical) representing the proportion of agreements in total votes.
         \end{itemize}

         \begin{table}[ht]
            \small
            \setlength{\tabcolsep}{0.4\tabcolsep}
            \begin{tabular}{lrrrrrrrrr}
               \toprule
               \multirow{2}{*}{Feature} & \mc{9}{c}{Network}                                                               \\
               \cmidrule{2-10}
               & \mc{3}{c}{Messages}     & \mc{3}{c}{US Flights}     & \mc{3}{c}{Voting}         \\
               \midrule
               Nodes                    & \mc{3}{c}{32}           & \mc{3}{c}{500}             & \mc{3}{c}{634}            \\
               Edges                    & \mc{3}{c}{266}          & \mc{3}{c}{2980}            & \mc{3}{c}{176426}         \\
               $S_D^{prod}$             & [ & -82.1,  & 82.1]     & [ & -40.0,     & 40.0]     & [ & -110.1,    & 110.1]   \\
               $S_C^{prod}$             & [ & -177.1, & 177.1]    & [ & -90.5,     & 90.5]     & [ & -58.4,     & 58.4]    \\
               $S_D^{sum} = S_C^{sum}$  & [ & -103.5, & 103.5]    & [ & -48.5,     & 48.5]     & [ & -104.7,    & 104.7]   \\
               $U_D$                    & [ & -108.6, & 104.1]    & [ & -147724.5, & 14114.7]  & [ & -12469.3,  & 7452.4]  \\
               $U_C$                    & [ & -177.8, & 114629.0] & [ & -86532.4,  & 119285.9] & [ & -160602.5, & 38498.1] \\
               \bottomrule
            \end{tabular}
            \caption{
               Information summary of the real-world networks used for experimentation.
               \label{tab:realdatasets}
            }
         \end{table}

         \Cref{fig:rewiring} depicts the assessment of the measures of interest on graphs derived from the just indicated networks.
         The former were generated by iterated rewiring (i.e., pairwise edge switching~\cite{milo2003uniform, opsahl2008prominence}) of the original definition of the latter:
         once again 1000 graphs were generated for each number of rewirings considered.
         According to past works, there is a smooth variation of the network characteristics as more rewires are performed.
         This goes until the resulting network can be considered plainly random and, in this sense, stable.
         This could be expected to occur when the number of repetitions is at least within the same order of magnitude of the number of edges.
         \begin{figure}[ht]
            \centering
            \includegraphics{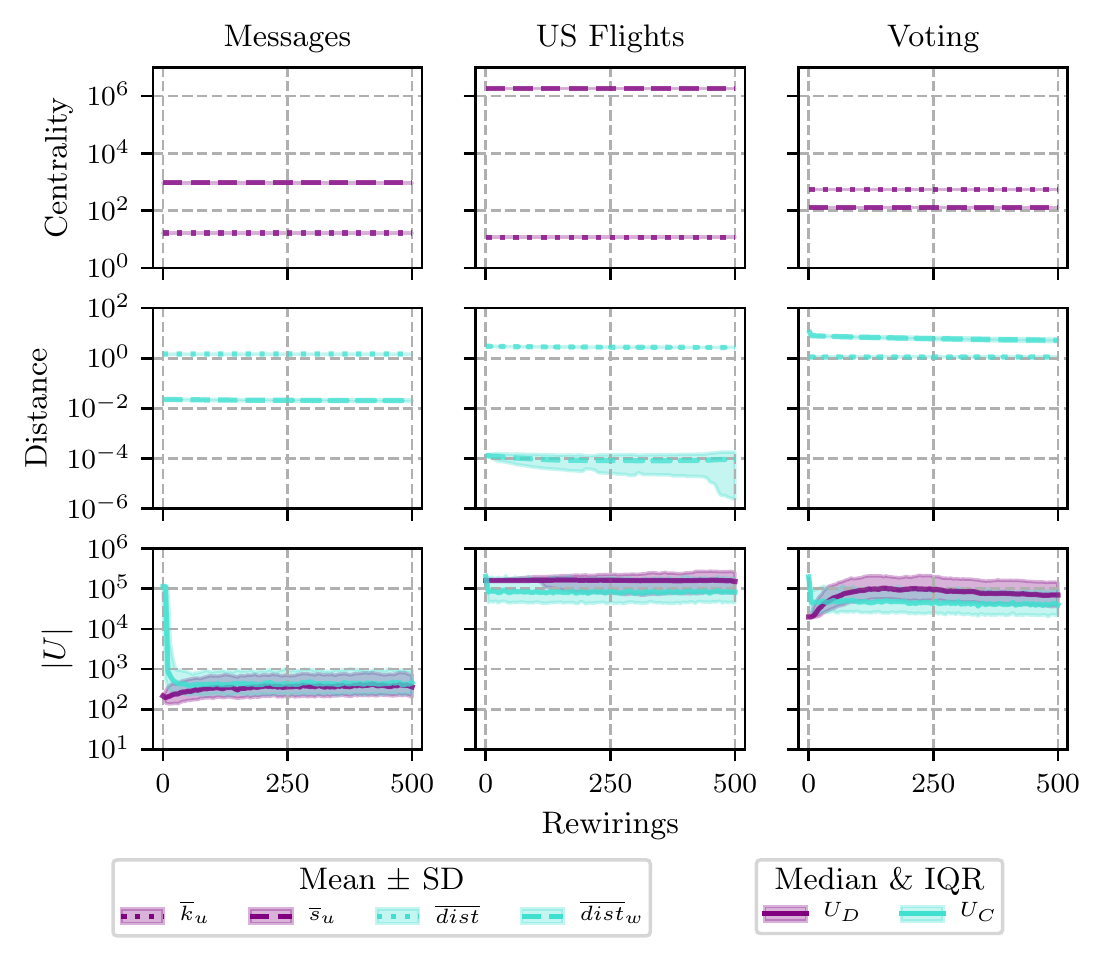}
            \caption{
               Evaluation of weights significance in real-world networks according to the increase randomization of their surrogates.
               \label{fig:rewiring}
            }
         \end{figure}

         While the invariance of degree and strength averages was assured by the rewiring procedure definition, that was not the case for distances.
         Nevertheless these were mildly affected, resembling the just mentioned invariance:
         only the `US Flights' surrogates clearly show a different behavior, as the dispersion of the weighted distances averages increases with the number of rewires.
         There is a similar consistency of $|U_D|$ and $|U_C|$, except after very few rewires:
         $|U_C|$ abruptly decreases in all 3 cases, implying a greater diversity of node closeness profiles;
         on the other hand, $|U_D|$ increases in a softer fashion for the `Messages' and `Voting' networks, evidencing more coinciding node strengths (since degrees are preserved).
         Both changes could be considered unexpected thanks to the small number of rewires which prompted them.

         Finally, \cref{fig:centrality_var} demonstrates one of the advantages of log-based centralities using the real-world networks concerned.
         Since the result reported for the exponential methods was submitted to the log function, what we see as a linear change in the variance of the centralities is, in fact, exponential.
         The log-based measure, however, displayed in linear scale, presents a quadratic behavior of the variance.
         That sounds natural, considering that as $\alpha \to +\infty$ or $-\infty$, the centralities defined according to the log method tend to move away from each other (assuming a bounded useful interval).
         That is not the case for the prod and sum alternatives since, in general, there is a discrepancy of scale in the centrality values comparing both extrema.
         \begin{figure}[ht]
            \centering
            \includegraphics{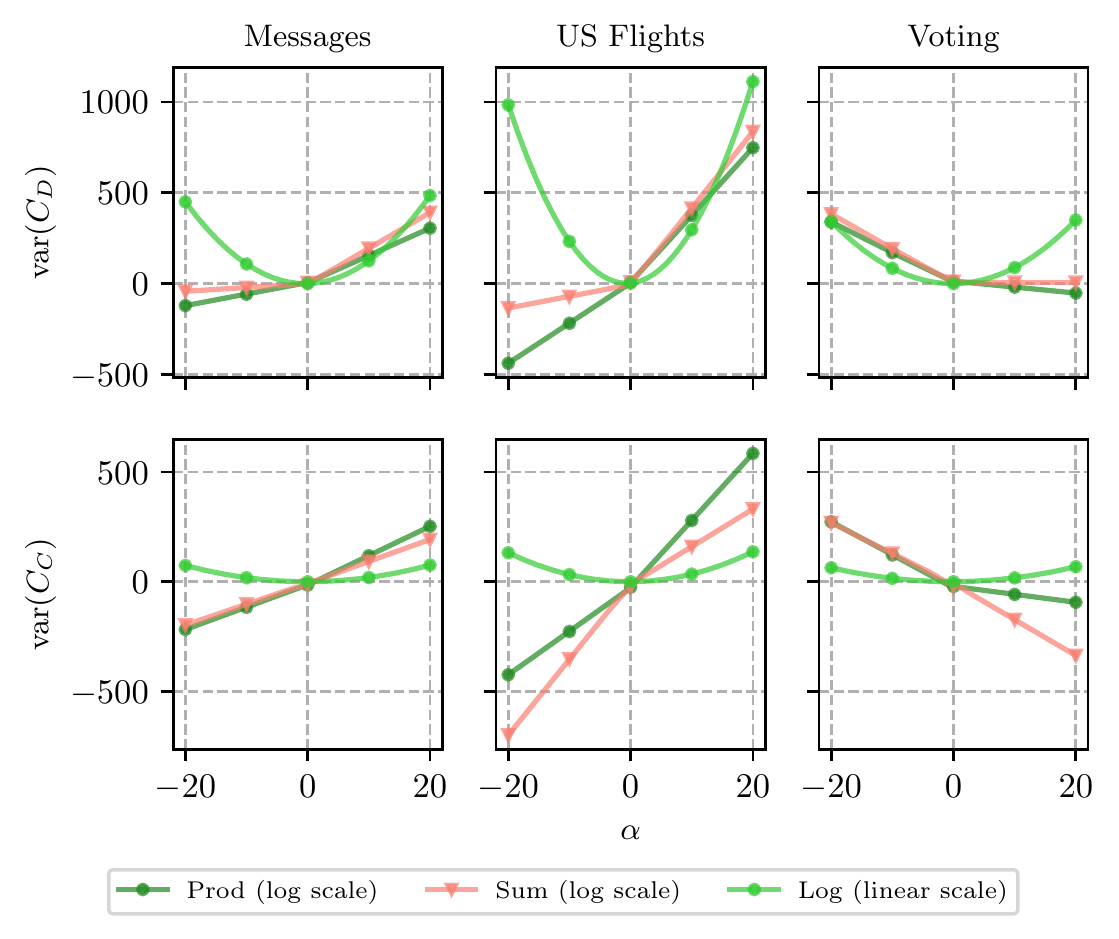}
            \caption{
               Comparison of the centralities variance of real-world networks.
               The effective values for the exponential methods can be obtained by taking $2^{var(C_D)}$ or $2^{var(C_C)}$.
               In practical terms, if the variance tends to grow abruptly on the one side, on the other it becomes almost null rapidly.
               The measures based on the proposer log method do not exhibit this behavior.
               \label{fig:centrality_var}
            }
         \end{figure}

         In addition, the figure also highlights some particularities, as the fact that the voting network presents decreasing variance for the exponential methods.
         This happens because the weights are within the unit interval and, therefore, unweighted degree and closeness exceed, or at least equate, their weighted version.
         The opposite is true for networks whose weights are greater or equal to one.
         Moreover, the curvature of the lines respective to the log method can be related to the length of safe intervals:
         $\text{var}(C_D)$ to $S_D^{prod}$, and $\text{var}(C_C)$ to $S_C^{sum}$.


   \section{Conclusion \label{sec:end}}
      The ubiquity of abstract and concrete systems which can be insightfully modeled as weighted networks is evident.
      This motivates the analysis of such structures from various points of view, such as node centrality.
      This work aimed at improving past ones on the consideration of edge weights for the establishment of the just mentioned perspective.
      Possibly the most repeated contribution in this sense is the proposal of measures to make the assessment in question.
      While this happens once again here, it is not an end in itself, but serves as the gateway to the entire methodology developed from there.
      Moreover, even this starting point is elaborately characterized and compared to its alternatives.

      The conception and the interpretation of the useful intervals is the major accomplishment reported in this text.
      They lead to an entirely unexplored way to analyze weighted networks by allowing to objectively evaluate the significance of edge weights for node centrality.
      And targeting to assure the completeness of this contribution, an elegant procedure for the computation of these intervals is also provided, which is flexible enough to tackle problems in other contexts with similar geometric properties.
      At last, the performed experiments successfully illustrate the behavior of the presented methods in various situations, corroborating their usefulness and uniqueness.

      About the next steps, we believe the application of the proposed methodology on a wider variety of real networks as well as probabilistically generated ones would enable the identification of other interesting patterns related to inherent characteristics of the former and parameter settings of the latter.
      Another future work could focus on the characterization of properties of well-known graph families based on the log-based centrality measures and their respective useful intervals.
      These are some promising potential continuations of this research.

   \section*{CRediT authorship contribution statement}
      We describe contributions to the paper as follows:
      R.S.P and D.O.C should be credited for conceptualization, methodology, software, validation, formal analysis, investigation, data curation, writing -- original draft, review and editing --, and visualization;
      D.O.C also should be credited for resources, supervision, project administration, and funding acquisition.

   \section*{Declaration of Competing Interest}
      The authors declare that they have no known competing financial interests or personal relationships that could have appeared to influence the work reported in this paper.

   \section*{Funding}
      This work was supported by the Carlos Chagas Filho Foundation for Research Support in Rio de Janeiro (FAPERJ) [grant number 248551 to R.S.P.].


   \bibliography{infosci2020_eancwn}
\end{document}